\documentclass[journal]{IEEEtran}
\usepackage[T1]{fontenc}

\ifCLASSINFOpdf
\else
\fi
\usepackage{amsmath}
\usepackage{amsthm}
\usepackage{amsfonts}
\usepackage{amssymb}
\usepackage{float}
\usepackage{pgfplots}
\usepackage{tikz}
\usepackage[colorinlistoftodos]{todonotes}
\usetikzlibrary[pgfplots.groupplots]

\newlength{\hsep}

\usepackage{fancyhdr}
\usepackage{lipsum}
\usepackage{fancyhdr}
\fancypagestyle{title}{%
  \setlength{\headheight}{22pt}%
  \fancyhf{}
  \lhead{PREPRINT, June 22, 2022}
\rhead{\thepage}
}%

\pgfplotsset{width=8.5cm,compat=1.17} 

\newtheorem{theorem}{Theorem}[]

\begin{document}
\title{Time-Limited Waveforms with Minimum Time Broadening for the Nonlinear Schr$\ddot{\text{o}}$dinger Channel}

\author{Youssef~Jaffal,  Alex~Alvarado, \textit{Senior Member, IEEE}

\thanks{This work has  received  funding  from  the  European  Research  Council (ERC) under the European Union’s Horizon 2020 research and innovation programme under Grant 757791. \textit{(Corresponding author: Youssef Jaffal.)}

The authors are with the Department of Electrical
Engineering, Eindhoven University of Technology, 5600 MB Eindhoven,
The Netherlands (e-mail: y.jaffal@tue.nl; a.alvarado@tue.nl).

}
}

\maketitle
\thispagestyle{title}
\begin{abstract}

Simple fiber optic communication systems can be implemented using energy modulation of isolated time-limited pulses. Fundamental solitons are one possible solution for such pulses which offer a fundamental advantage: their shape is not affected by fiber disperison and nonlinearity. Furthermore, a simple energy detector can be used at the receiver to detect the transmitted information. However, systems based on energy modulation of solitons are not competitive in terms of data rates. This is partly due to 
the fact that the effective time duration of a soliton depends on its chosen amplitude. In this paper, we propose to replace fundamental solitons by new time-limited waveforms that can be detected using an energy detector, and that are immune to fiber distortions. Our proposed solution relies on the prolate spheroidal wave functions and a numerical optimization routine. Time-limited waveforms that undergo minimum time broadening along an optical fiber are obtained and shown to outperform fundamental solitons. In the case of binary transmission and a single span of fiber, we report rate increases of $33.8$\% and $12$\% over lossy and lossless fibers, respectively. Furthermore, we show that the transmission rate of the proposed system increases as the number of used energy levels increases, which is not the case for fundamental solitons due to their effective time-amplitude constraint. For example,  rate increases of $164$\%  and $70$\% over lossy and lossless fibers respectively are reported when using four energy levels.

\end{abstract}


%
\IEEEpeerreviewmaketitle

\section{Introduction}
\label{Sect1}

The propagation of signals through a single-mode optical fiber is governed by the nonlinear Schr$\ddot{\text{o}}$dinger (NLS) equation \cite[Chapter 2]{AGRAWAL201927}. Communication systems for the NLS channel can be divided into two main categories. The first one uses linear modulation techniques combined with either linear or nonlinear (e.g., digital backpropagation) receivers. Such communication systems (with linear receivers) are currently implemented in most commercial systems. The second category includes systems with nonlinear transmitters and receivers, e.g., those based on the nonlinear Fourier transform (NFT). Such systems are tailored to the NLS equation \cite{yousefi2014Part3} and are currently still at a research stage. The simplest NFT-based transceiver uses energy modulation (EM) of isolated fundamental solitons, where the inverse NFT is not required at the receiver and a simple energy detector can be used. Such systems date in fact back to 1973, where the idea of using solitons for transmission was first proposed by Hasegawa \cite{hasegawa1973transmission}.

The transmission rates of EM of isolated fundamental solitons are lower than the ones achieved by other sophisticated techniques (such as current coherent systems or modulation techniques that need the NFT or its inverse). Systems based on EM of isolated solitons suffer from several weaknesses.
First, EM systems have limited transmission rates as they limit the number of useful signals. For example, the energy detector cannot differentiate between two different signals that have the same energy but different phases.  
Second, fundamental solitons are not time-limited, and hence consecutive fundamental solitons are not perfectly isolated. Generating a time-limited signal from a fundamental soliton can be done by simple truncation or by more sophisticated methods such as the soliton shortening methods recently introduced in \cite{Sander2021}. Increasing the time window of the truncated or shortened soliton reduces the distortion and provides better isolation between the consecutive pulses along the NLS channel. Such time increase  decreases the bit error rate, but it also decreases the transmission rates. Third, fundamental solitons have a high time-bandwidth product \cite{span2019time}, which results in a low spectral efficiency. 
Fourth,  low energy solitons need relatively high time duration since the energy and the effective duration of solitons are inversely proportional. Consequently, conventional pulse amplitude modulation (PAM) systems with more than two amplitude levels have lower performance compared to an on-off keying (OOK) system as reported in \cite[Fig.~2(d)]{chen2020capacity}. 

 Fundamental solitons are not the only candidate pulses to be used in EM of isolated pulses; as the transmitted information is embedded in the energy of the transmitted pulses, preserving the shape of the transmitted pulses along the NLS channel is in fact not required. The desirable property of the used pulses is to have low time  broadening: large time broadening would require a large guard time to guarantee that consecutive pulses remain isolated along the channel, which would decrease the transmission rates of the system. Compared to a fundamental soliton (that preserves its effective duration), a shorter pulse with the same energy is more desirable if the effective duration of the broadened received pulse remains lower than the effective duration of the received soliton.

\begin{figure*}[t!]
\centering
\input{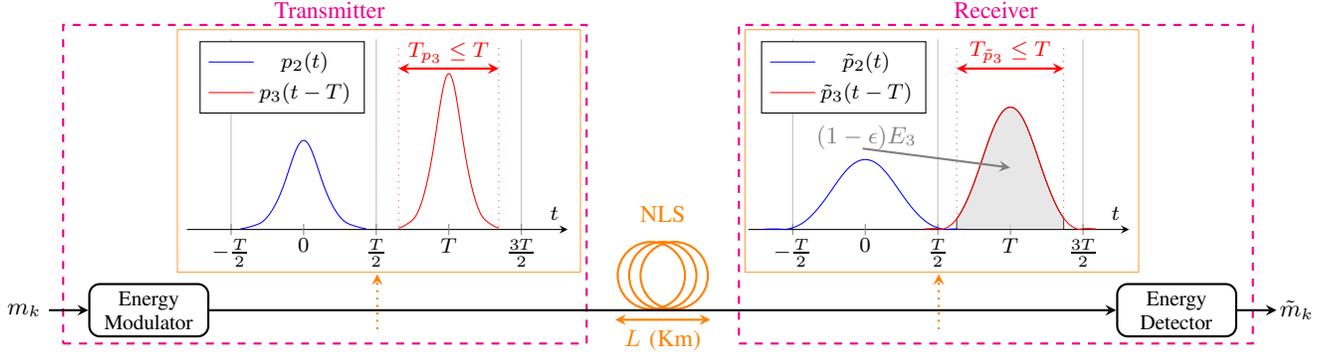}
\caption{Energy modulation of isolated pulses, where consecutive pulses carry the information bits in their energies. $T$ should be greater or equal to the duration of each transmitted pulse ($T_{p_m}$), and greater or equal to the duration that contains $100(1-\epsilon)\%$ of the energy of each received pulse  ($T_{\tilde{p}_m}$).}
    \label{FigSystemModel}
\end{figure*}
 
 In this paper we propose replacing fundamental solitons by time-limited waveforms that undergo  \textit{minimum time broadening (MTB)} in EM of isolated pulses.
We make use of the set of prolate spheroidal wave functions (PSWFs) to search for the MTB waveforms. The PSWFs form a complete orthonormal set for time-limited functions and they were used in \cite{halpern1979optimum}, \cite{jaffal2019achievable}, and \cite{jaffal2022time} to design optimal time-limited pulses. We consider both lossless and lossy fibers 
and we show that the obtained MTB waveforms have lower time duration and lower time-bandwidth product compared to fundamental solitons. Consequently, MTB waveforms outperform the truncated solitons in conventional EM of isolated pulses. More importantly, we show that MTB waveforms have the ability to increase the transmission rates and the spectral efficiency when additional low energy pulses are used, which as previously mentioned, is not the case for solitons. The main contribution of this paper is introducing MTB waveforms that improve the performance of EM of isolated pulses in fiber optic channels.

The paper is organized as follows: in Sec. \ref{Sect:SystMod} we present the considered system model. In Sec. \ref{Sec:Solitons} we present the numerical results  of using fundamental solitons in EM of isolated pulses. In Sec. \ref{Sec:MTBWaveforms} we propose the use of MTB waveforms and show their advantage over fundamental solitons. We  conclude the paper in Sec. \ref{Sect:conc}.

\section{System Model: Energy Modulation of Isolated Pulses}
\label{Sect:SystMod}

We consider the EM system illustrated in Fig. \ref{FigSystemModel}. The transmitted messages $\{m_k\}_{k\in\mathbb{Z}}$ are chosen from the set $\mathcal{M}=\{1,2, \dots, M\}$ where we restrict $M$ to be a power of two. For each message $m\in \mathcal{M}$, the energy modulator generates a corresponding pulse $p_m(t)$ that has an energy 
\begin{equation}
E_{m} = \left(\frac{m-1 }{M-1}\right)^2E_{M},
\label{Eq:EnergiesEM}
\end{equation}
where $E_{M}$ is the highest used energy level. Note that if the chosen pulses $p_m(t)$ are scaled versions of the same pulse, i.e.,  $p_m(t)=\frac{m-1}{M-1}p_{M}(t)$ $\forall m\in\mathcal{M}$, then the considered EM system is equivalent to conventional PAM with equispaced amplitudes, where the transmitted information bits are embedded in the amplitude of the same pulse. 

We consider the use of strictly time-limited pulses. We define the duration $T_{p_m}$ of $p_m(t)$ by the minimum time window that satisfies
\begin{equation}
       p_m(t)=0 \text{ if } t \notin \left[-\frac{T_{p_m}}{2}, \frac{T_{p_m}}{2} \right].
    \label{eq:pulsePtl}
\end{equation}
Note that $T_{p_1}=0$ since $p_1(t)$ is the zero-energy signal (see (\ref{Eq:EnergiesEM})). 
The energy modulator transmits the signal 
\begin{equation}
q(t,0)=\sum_{k\in\mathbb{Z}}p_{m_k}(t-kT),
\label{eq:TxSignal}
\end{equation}
where $T$ is the modulation interval.
To guarantee strict isolation between the consecutive transmitted pulses, we restrict $T$ to be greater or equal to the duration of the pulses, i.e., 
\begin{equation}
   T\geq T_{p_m}   \text{ for } m\in\{2,3,\dots, M\}.
   \label{eq:TxIsolation}
\end{equation}

Because of the time limitation in \eqref{eq:pulsePtl}, $\left\{p_{m}(t)\right\}_{m=2}^{M}$ cannot be strictly band-limited.  We define their effective bandwidth by the values of $W_{p_m}$ that satisfy
\begin{equation}
       \int_{-W_{p_m}}^{W_{p_m}} \left|P_m(f)\right|^2df = (1-\epsilon)E_m  \text{ for } m\in\{2,3,\dots, M\},
       \label{eq:Banddef}
\end{equation}
where $P_m(f)=\int_{-\infty}^{\infty}p_m(t)e^{-j2\pi ft}dt$ is the Fourier transform of $p_m(t)$, and $\epsilon E_m$ is the corresponding out-of-band energy where $0<\epsilon<1$. Note that $W_{p_1}=0$ since $p_1(t)$ is the zero signal. We consider a maximum allowable bandwidth $W_{max}$ where the bandwidths of the transmitted pulses should satisfy
\begin{equation}
       W_{p_m}\leq W_{max} \text{ for } m\in\{2,3,\dots, M\}.
       \label{eq:BandConstr}
\end{equation}

The considered channel is a single-mode optical fiber of length $L$, where the evolution of the propagated signal $q(t,z)$ is governed by the NLS equation \cite[eq. (2.3.46)]{AGRAWAL201927}
\begin{equation}
    \frac{\partial q(t,z)}{\partial z} = -\frac{\alpha}{2}q(t,z) -\frac{i\beta_2}{2}\frac{\partial^2 q(t,z)}{\partial^2 t} + i\gamma|q(t,z)|^2q(t,z),
    \label{NLSEq}
\end{equation}
where $t$ represents retarded time, $z$ is the propagation distance, $\alpha$ represents the fiber losses, $\beta_2$ is the group-velocity dispersion parameter, and $\gamma$ is the nonlinear parameter. 
The NLS equation is used to describe two main models for fiber optical channels. The first model is when using ideal distributed amplification where we assume that the fiber losses are ideally compensated, hence $\alpha=0$. We refer to this model as ``lossless fiber''.
The second one is the ``lossy fiber'' model where an erbium-doped fiber amplifier (EDFA) is used at the end of the optical fiber, so $\alpha$ is non zero along the fiber. 
Both amplification schemes add noise to the system, but in this work we ignore the noise and we focus on the deterministic pulse propagation in  optical fibers.

The receiver can obtain the squared magnitude of the propagated signal $|q(t,L)|^2$ using a photodiode detector for example. 
Let 
\begin{equation}
  \left\{\Tilde{p}_{m}(t)\right\}_{m\in \mathcal{M}}  
  \label{eq:RxPulses}
\end{equation}
to be the magnitude of the propagated $\left\{p_{m}(t)\right\}_{m\in \mathcal{M}}$, i.e.,   $\forall m\in\mathcal{M}$
\begin{equation*}
\Tilde{p}_{m}(t) = |q(t,L)|    \text{ if }q(t,0) = p_{m}(t).
\label{}
\end{equation*}
We assume that the energy of a received pulse is equal to the energy of the corresponding transmitted one, which can be achieved by using amplifiers. Hence, $E_m$ is the energy of $\Tilde{p}_m(t)$ $\forall m\in \mathcal{M}$. The time support of the propagated pulses may change along the NLS channel, and we refer to this effect as ``time broadening''. Note that $\Tilde{p}_1(t)$ is the zero signal and has zero duration and zero bandwidth.
We define the effective duration of $\Tilde{p}_m(t)$ with $m\in\{2,3, \dots, M\}$ by the time window  $T_{\Tilde{p}_m}$ that satisfies 
\begin{equation}
    \int_{-T_{\Tilde{p}_m}/2}^{T_{\Tilde{p}_m}/2} \left|\Tilde{p}_m(t)\right|^2dt = (1-\epsilon)E_m.
    \label{eq:Teff}
\end{equation}

To control the interaction between consecutive received pulses, we impose the following constraint on the modulation interval:
\begin{equation}
    T\geq T_{\tilde{p}_m}  \text{ for } m\in\{2,3,\dots, M\}.
    \label{eq:RxIsolation}
\end{equation}
Decreasing $\epsilon$ improves the isolation between propagated consecutive pulses, hence, consecutive pulses are said to be \textit{isolated at level $1-\epsilon$} at the receiver. Note that the isolation should be also maintained along the channel and not only at the receiver. If the effective duration of the propagated pulses along the channel does not exceed the effective duration at the receiver, then the constraint in (\ref{eq:RxIsolation}) is sufficient. Otherwise, as in the case of multi-solitons \cite{span2019time}, the constraint in (\ref{eq:RxIsolation}) should be extended to take into account the effective duration along the channel. 

By choosing a sufficiently small value for $\epsilon$ in (\ref{eq:Teff}), the interaction between the possibly broadened consecutive pulses can be kept negligible when transmitting the signal in (\ref{eq:TxSignal}) through \eqref{NLSEq}. Then, the received signal $|q(t,L)|$ is well approximated by
\begin{equation}\label{aaaaa}
|q(t,L)|\approx\sum_{k\in\mathbb{Z}} \Tilde{p}_{m_k}(t-kT).
\end{equation}
 In this case, the energy detector can generate the estimate $\{\Tilde{m}_k\}$ of $\{m_k\}$ by computing the energy over the time window $[kT-T/2, \, kT+T/2]$. 

 The transmission rate of this system is
\begin{equation}
    R=\frac{\log_2M}{T}\text{~bit/s}.
    \label{eq:rates}
\end{equation}
We consider the use of the modulation interval
\begin{equation}
    T = \max\left\{T_{p_2}, T_{\Tilde{p}_2},T_{p_3}, T_{\Tilde{p}_3},\dots, T_{p_{M}}, T_{\Tilde{p}_{M}}\right\}.
    \label{eq:Tmax}
\end{equation}
To satisfy the constraints in (\ref{eq:TxIsolation}) and (\ref{eq:RxIsolation}), the modulation interval should be greater or equal to the right hand side of \eqref{eq:Tmax}. Hence, the considered modulation interval in \eqref{eq:Tmax} is the one that maximizes \eqref{eq:rates} while satisfying  (\ref{eq:TxIsolation}) and (\ref{eq:RxIsolation}).
We use the bandwidth definition in (\ref{eq:Banddef}) for the pulses $p_m(t)$ and $\Tilde{p}_m(t)$, and we define the spectral efficiency by 
\begin{equation}
    \text{SE} \triangleq \frac{R}{W}=\frac{\log_2M}{WT} \text{~bit/s/Hz},
    \label{eq:SE}
\end{equation}
where 
\begin{equation}
    W=\max\left\{W_{p_2}, W_{\Tilde{p}_2},W_{p_3}, W_{\Tilde{p}_3},\dots, W_{p_{M}}, W_{\Tilde{p}_{M}}\right\}.
    \label{eq:WeqMax}
\end{equation}
If the effective bandwidths along the channel do not exceed the ones at the transmitter and the receiver, then \eqref{eq:SE} and \eqref{eq:WeqMax} are valid definitions. Otherwise, the effective bandwidth along the channel should be also considered as done in \cite{span2019time} for multi-solitons.
We define the time-bandwidth product by
\begin{equation}
    c\triangleq WT,
    \label{eq:TimBandProd}
\end{equation}
which is inversely proportional to the spectral efficiency: the lower the $c$, the higher the spectral efficiency.

In the following section we consider the truncated solitons as a baseline for our system model, and in Sec. \ref{Sec:MTBWaveforms} we propose the MTB waveforms and present their numerical results. In what follows we use the following simulation parameters: $L=80$~Km, $\epsilon=10^{-4}$, $W_{max}=50$~GHz, $\beta_2=-21.7$~$\text{ps}^2$/Km, $\gamma=1.2$~$\text{W}^{-1}\text{Km}^{-1}$, $\alpha=0$~dB/Km for the lossless fiber, and $\alpha=0.2$~dB/Km for the lossy fiber.

\section{Energy Modulation of truncated solitons}
\label{Sec:Solitons}

In this section, we first present the definition and some properties of fundamental solitons, then we present their use in our system model.
In this paper, we use the following equation for fundamental solitons \cite[Sec. 5.2.2]{AGRAWAL201927}:
\begin{equation}
    s(t)=A \, \text{sech}\left(A\sqrt{\frac{\gamma}{|\beta_2|}} t \right),
    \label{eq:soliton}
\end{equation}
which is a solution for (\ref{NLSEq}) when $\alpha=0$, i.e., if $q(t,0)=s(t)$ is used in (\ref{NLSEq}), then $\tilde{s}(t)=|q(t,z)|=s(t)$ $\forall z\in \mathbb{R}$.
The energy of $s(t)$ is 
\begin{equation}
    E_{s}= 2A\sqrt{\frac{|\beta_2|}{\gamma}},
    \label{ESoliton}
\end{equation}
and its effective duration is given by
\begin{equation}
T_{s} = \frac{\sqrt{|\beta_2|}}{A\sqrt{\gamma}}\ln{\frac{2-\epsilon}{\epsilon}},
\label{TeffSoliton}
\end{equation}
which is obtained by solving (\ref{eq:Teff}) for $\tilde{p}_m(t)=s(t)$ (similarly to (9) in \cite{chen2020capacity} which was derived in the case of normalized NLS equation).

The Fourier transform of $s(t)$ is \cite[Appendix F]{lamb1995introductory}
\begin{equation*}
    S(f)= \pi\sqrt{\frac{|\beta_2|}{\gamma}} \, \text{sech}\left(\frac{\pi^2}{A}\sqrt{\frac{|\beta_2|}{\gamma}} f \right).
\end{equation*}
Its effective bandwidth is obtained by solving (\ref{eq:Banddef}) for $P_m(f)=S(f)$ and it is given by
\begin{equation}
    W_{s}= \frac{A\sqrt{\gamma}}{\pi^2\sqrt{|\beta_2|}}\ln{\frac{2-\epsilon}{\epsilon}}.
    \label{WeffSoliton}
\end{equation}
By (\ref{TeffSoliton}) and (\ref{WeffSoliton}), the time-bandwidth product of fundamental solitons is given by
\begin{equation}
    c_{s}= T_s W_s = \frac{1}{\pi^2}\ln^2{\frac{2-\epsilon}{\epsilon}}.
    \label{cSoliton}
\end{equation}

By (\ref{ESoliton}) and (\ref{WeffSoliton}),  the energy of the soliton is directly proportional to its effective bandwidth. Hence using a maximum constraint on the used bandwidth $W_{max}=50$~GHz yields a maximum possible energy for fundamental solitons of approximately $1.8$~pJ (for the considered values of $\epsilon$, $\beta_2$, and $\gamma$). 
We consider truncated solitons as a baseline of our system, where a soliton is truncated to its effective duration $T_{s}$ (see (\ref{TeffSoliton})). Given an energy level $E$, we generate the soliton over its effective duration using (\ref{eq:soliton}), (\ref{ESoliton}), and (\ref{TeffSoliton}). We obtain the received truncated soliton using the split-step Fourier method (SSFM) \cite[Sec. 2.4.1]{AGRAWAL201927}, and we estimate numerically its effective duration $T_{\tilde{s}}$. In Fig. \ref{FigTsolitons} we show the transmitted duration (solid) and the estimated received duration (dashed) for both lossless and lossy fibers as a function of the energy $E$. Note that $T_s$ and $E$ are inversely proportional (by (\ref{ESoliton}) and (\ref{TeffSoliton})), where the minimum $T_s$ is obtained at the maximum energy $E=1.8$~pJ, and where $T_{s}\to \infty$ as $E\to 0$. For energies that are lower than $0.9$~pJ, $T_{\tilde{s}}$ is almost the same for lossy and lossless fibers (no broadening). $T_{\tilde{s}}$ is slightly lower than $T_{s}$, then $\max\{T_s, T_{\tilde{s}}\}=T_{s}$. For energies higher than $0.9$~pJ, $T_{\tilde{s}}$ is larger than $T_{s}$, and hence $\max\{T_s, T_{\tilde{s}}\}=T_{\tilde{s}}$ for both lossy and lossless fibers.  We verified that the effective duration of the truncated solitons along the fibers does not exceed their effective duration at the receiver, hence the constraint in \eqref{eq:RxIsolation} is sufficient to guarantee isolation along both the channel and the receiver. 
 
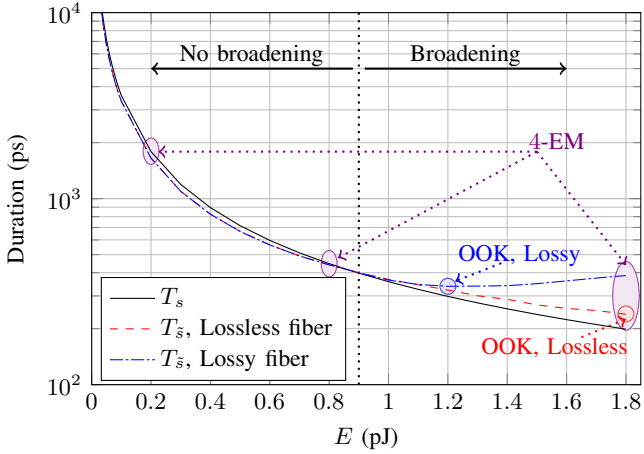
\begin{figure}[t!]
\centering
\begin{tikzpicture}[font=\small,
    Gclabel/.style={text=black},]
\begin{semilogyaxis}[
    xlabel={$E$ (pJ)},
    ylabel={Duration (ps)},
    xmin=0.0, xmax=1.85,
    ymin=100, ymax=10000,
    xtick={0,0.2,0.4,0.6,0.8, 1, 1.2,1.4, 1.6,1.8},
    ytick={0.1, 1, 10, 100, 1000, 10000, 100000, 1000000, 10000000},
    legend style={at={(0.24,0.01)},anchor=south},
    ymajorgrids=true,
    xmajorgrids =true,
    yminorgrids=true,
    xminorgrids =true,
    legend cell align={left},    
    scale only axis,
    width=7.3cm,
    height=4.95cm,
]
\filldraw[color=violet!75, fill=violet!10](1.8,300) ellipse  (5pt and 13pt);
\draw[violet, ->, dotted, thick] (1.5,1789) -- (1.8,440);

\filldraw[color=violet!75, fill=violet!10](0.8,447) ellipse  (3pt and 5pt);
\draw[violet, ->, dotted, thick] (1.5,1789) -- (0.83,480)node[pos=-0.1] {$4$-EM};

\filldraw[color=violet!75, fill=violet!10](0.2,1800) ellipse  (3pt and 5pt);
\draw[violet, ->, dotted, thick] (1.5,1789) -- (0.225,1789);

\filldraw[color=red!75, fill=red!10](1.8,240) ellipse  (3pt and 3pt);
\draw[red, ->, dotted, thick] (1.65,180) -- (1.8,225) node[pos=-0.6] {OOK, Lossless};

\filldraw[color=blue!75, fill=blue!10](1.2,338) ellipse  (3pt and 3pt);
\draw[blue, ->, dotted, thick] (1.38,450) -- (1.225,360) node[pos=-0.4] {OOK, Lossy};

\draw[black, -, thick, dotted] (0.9,100) -- (0.9,10000);
\draw[black, ->,  thick] (0.92,5000) -- (1.6,5000)node[pos=0.5,yshift=5] {Broadening};
\draw[black, ->,  thick] (0.88,5000) -- (0.2,5000)node[pos=0.5,yshift=5] {No broadening};

\addplot[
    color=black,
    ]
    coordinates {
(9.999010e-06, 3.578874e+07)(4.999505e-05, 7.157749e+06)(9.999010e-05, 3.578874e+06)(1.999802e-04, 1.789437e+06)(4.999505e-04, 7.157749e+05)(7.999208e-04, 4.473593e+05)(9.999010e-04, 3.578874e+05)(1.999802e-03, 1.789437e+05)(4.999505e-03, 7.157749e+04)(7.999208e-03, 4.473593e+04)(9.999010e-03, 3.578874e+04)(1.999802e-02, 1.789437e+04)(2.999703e-02, 1.192958e+04)(3.999604e-02, 8.947186e+03)(4.999505e-02, 7.157749e+03)(5.999406e-02, 5.964791e+03)(6.999307e-02, 5.112678e+03)(7.999208e-02, 4.473593e+03)(8.999109e-02, 3.976527e+03)(9.999010e-02, 3.578874e+03)(1.999802e-01, 1.789437e+03)(2.999703e-01, 1.192958e+03)(3.999604e-01, 8.947186e+02)(4.999505e-01, 7.157749e+02)(5.999406e-01, 5.964791e+02)(6.999307e-01, 5.112678e+02)(7.999208e-01, 4.473593e+02)(8.999109e-01, 3.976527e+02)(9.999010e-01, 3.578874e+02)(1.099891e+00, 3.253522e+02)(1.199881e+00, 2.982395e+02)(1.299871e+00, 2.752980e+02)(1.399861e+00, 2.556339e+02)(1.499851e+00, 2.385916e+02)(1.599842e+00, 2.236797e+02)(1.699832e+00, 2.105220e+02)(1.799822e+00, 1.988264e+02)

    };
\addlegendentry{$T_{s}$}

\addplot[
    color=red,
    dashed,
    ]
    coordinates {
(9.999010e-06, 3.327235e+07)(4.999505e-05, 6.654470e+06)(9.999010e-05, 3.327235e+06)(1.999802e-04, 1.663617e+06)(4.999505e-04, 6.654470e+05)(7.999208e-04, 4.159044e+05)(9.999010e-04, 3.327235e+05)(1.999802e-03, 1.663617e+05)(4.999505e-03, 6.654470e+04)(7.999208e-03, 4.159044e+04)(9.999010e-03, 3.327235e+04)(1.999802e-02, 1.663617e+04)(2.999703e-02, 1.109078e+04)(3.999604e-02, 8.318087e+03)(4.999505e-02, 6.654470e+03)(5.999406e-02, 5.545391e+03)(6.999307e-02, 4.753193e+03)(7.999208e-02, 4.159044e+03)(8.999109e-02, 3.696928e+03)(9.999010e-02, 3.334225e+03)(1.999802e-01, 1.649637e+03)(2.999703e-01, 1.095098e+03)(3.999604e-01, 8.248187e+02)(4.999505e-01, 6.668450e+02)(5.999406e-01, 5.638591e+02)(6.999307e-01, 4.922950e+02)(7.999208e-01, 4.403693e+02)(8.999109e-01, 3.999827e+02)(9.999010e-01, 3.683724e+02)(1.099891e+00, 3.444158e+02)(1.199881e+00, 3.227045e+02)(1.299871e+00, 3.005695e+02)(1.399861e+00, 2.875881e+02)(1.499851e+00, 2.684156e+02)(1.599842e+00, 2.568821e+02)(1.699832e+00, 2.504061e+02)(1.799822e+00, 2.392130e+02)

    };
\addlegendentry{$T_{\tilde{s}}$, Lossless fiber}

\addplot[
    color=blue,
    dashed,
    dash pattern=on 5pt off 1pt on 1pt off 1pt
    ]
    coordinates {
(9.999010e-06, 3.327235e+07)(4.999505e-05, 6.654470e+06)(9.999010e-05, 3.327235e+06)(1.999802e-04, 1.663617e+06)(4.999505e-04, 6.654470e+05)(7.999208e-04, 4.159044e+05)(9.999010e-04, 3.327235e+05)(1.999802e-03, 1.663617e+05)(4.999505e-03, 6.654470e+04)(7.999208e-03, 4.159044e+04)(9.999010e-03, 3.327235e+04)(1.999802e-02, 1.663617e+04)(2.999703e-02, 1.109078e+04)(3.999604e-02, 8.318087e+03)(4.999505e-02, 6.654470e+03)(5.999406e-02, 5.545391e+03)(6.999307e-02, 4.753193e+03)(7.999208e-02, 4.159044e+03)(8.999109e-02, 3.696928e+03)(9.999010e-02, 3.334225e+03)(1.999802e-01, 1.649637e+03)(2.999703e-01, 1.095098e+03)(3.999604e-01, 8.248187e+02)(4.999505e-01, 6.668450e+02)(5.999406e-01, 5.638591e+02)(6.999307e-01, 4.922950e+02)(7.999208e-01, 4.403693e+02)(8.999109e-01, 3.999827e+02)(9.999010e-01, 3.648774e+02)(1.099891e+00, 3.456867e+02)(1.199881e+00, 3.378495e+02)(1.299871e+00, 3.387456e+02)(1.399861e+00, 3.405123e+02)(1.499851e+00, 3.490335e+02)(1.599842e+00, 3.612951e+02)(1.699832e+00, 3.741700e+02)(1.799822e+00, 3.860027e+02)

    };
\addlegendentry{$T_{\tilde{s}}$, Lossy fiber}

\end{semilogyaxis}

\end{tikzpicture}
    \caption{Duration of transmitted ($T_s$) and received ($T_{\tilde{s}}$) truncated solitons for $L=80$~Km and $\epsilon=10^{-4}$.}
    \label{FigTsolitons}
\end{figure}

The transmission rate $R$ in \eqref{eq:rates} of EM of truncated solitons can be computed based on the results of Fig. \ref{FigTsolitons}. The transmission rate of an OOK system is maximized when the pulse with minimum $T=\max\{T_s, T_{\tilde{s}}\}$ is used. As shown in Fig. \ref{FigTsolitons}, the minimum $T$ for OOK is $338$~ps for lossy fibers (at $E=1.2$~pJ), and $239$~ps for lossless fibers (at $E=1.8$~pJ). We compute the corresponding transmission rates and we include them in the first column of Table~\ref{TableRates}. Note that using the truncated soliton with $E=1.8$~pJ over lossy fiber decreases the transmission rate by $12.5$\% compared to using the one with $E=1.2$~pJ, since $T_{\tilde{s}}=386$~ps at $E=1.8$~pJ.  

Consider now the case with $M=4$ energies. The energy levels that satisfy \eqref{Eq:EnergiesEM} and yield the minimum possible $T$ (defined in \eqref{eq:Tmax}) are $\{0, 0.2, 0.8, 1.8\}$~pJ for both lossy and lossless fibers (as shown in Fig. \ref{FigTsolitons}). 
In Fig. \ref{Fig4EMSolitons} we show the truncated solitons for the different energy levels, where the duration of a soliton increases as the energy decreases. The modulation interval in this case is $T=1789$~ps which is the same for lossy and lossless fibers. The corresponding transmission rates are computed and included in the first column of Table \ref{TableRates}. Note that the transmission rates for $4$-EM is lower than the ones achieved by OOK (similarly to the results in \cite[Fig.~2(d)]{chen2020capacity}). In the following theorem we derive an upper bound for the transmission rate, which decays to zero as $M\to\infty$.

\begin{theorem}
\label{Th:UpperBound}
Given a maximum allowable bandwidth $W_{max}$, the transmission rate of $M$-EM of isolated solitons is upper-bounded by
\begin{equation*}
    R \leq \frac{\pi^2W_{max}\log_2 M}{(M-1)^2 \ln^2{\frac{2-\epsilon}{\epsilon}}} \text{ bit/s.}
\end{equation*}
\end{theorem}
\begin{proof}
In the considered system model, $T_{p_2}$ and $T_{p_M}$ are the effective duration of the solitons with energies $E_2$ and $E_M$, respectively. $E_2=\frac{E_{M}}{(M-1)^2}$ by \eqref{Eq:EnergiesEM}, and since the effective duration and the energy of solitons are inversely proportional (by \eqref{ESoliton} and \eqref{TeffSoliton}), $T_{p_2}=(M-1)^2T_{p_M}$. By \eqref{cSoliton}, 
\begin{equation}
T_{p_2}=(M-1)^2T_{p_M}\geq \frac{(M-1)^2 \ln^2{\frac{2-\epsilon}{\epsilon}}}{\pi^2W_{max}}.
\label{eq:ThProofTp2}
\end{equation} 
By \eqref{eq:TxIsolation}, 
\begin{equation*}T\geq T_{p_2}\geq \frac{(M-1)^2 \ln^2{\frac{2-\epsilon}{\epsilon}}}{\pi^2W_{max}}.
\end{equation*} 
Finally, using \eqref{eq:rates} and the obtained lower bound on $T$ concludes the proof. 
\end{proof}

\begin{figure}[t!]
\centering
\begin{tikzpicture}[font=\small,
    Gclabel/.style={text=black},]
\begin{semilogyaxis}[
    xlabel={$t$ (ps)},
    xmin=-1000, xmax=1000,
    ymin=0.0, ymax=0.22,
    xtick={-900,-600,-300,0,300,600,900},
    legend style={at={(0.8,0.99)},anchor=north},
    ymajorgrids=true,
    xmajorgrids =true,
    scale only axis,
    width=7.2cm,
    height=3.75cm,
]

\draw[violet, <->,  thick] (-894.5,0.0007) -- (894.5,0.0007)node[pos=0.5, below] {$T=1789$\,ps};

\addplot[
    color=black,
    thick,
    dotted,
    ]
    coordinates {
(-9.941318e+01,2.994204e-03)(-9.747152e+01,3.298214e-03)(-9.552985e+01,3.633082e-03)(-9.358819e+01,4.001939e-03)(-9.164652e+01,4.408229e-03)(-8.970486e+01,4.855748e-03)(-8.776320e+01,5.348673e-03)(-8.582153e+01,5.891600e-03)(-8.387987e+01,6.489591e-03)(-8.193821e+01,7.148215e-03)(-7.999654e+01,7.873598e-03)(-7.805488e+01,8.672478e-03)(-7.611322e+01,9.552265e-03)(-7.417155e+01,1.052110e-02)(-7.222989e+01,1.158793e-02)(-7.028822e+01,1.276258e-02)(-6.834656e+01,1.405583e-02)(-6.640490e+01,1.547947e-02)(-6.446323e+01,1.704646e-02)(-6.252157e+01,1.877093e-02)(-6.057991e+01,2.066833e-02)(-5.863824e+01,2.275549e-02)(-5.669658e+01,2.505070e-02)(-5.475492e+01,2.757380e-02)(-5.281325e+01,3.034619e-02)(-5.087159e+01,3.339089e-02)(-4.892992e+01,3.673251e-02)(-4.698826e+01,4.039713e-02)(-4.504660e+01,4.441218e-02)(-4.310493e+01,4.880613e-02)(-4.116327e+01,5.360808e-02)(-3.922161e+01,5.884709e-02)(-3.727994e+01,6.455130e-02)(-3.533828e+01,7.074675e-02)(-3.339661e+01,7.745567e-02)(-3.145495e+01,8.469445e-02)(-2.951329e+01,9.247092e-02)(-2.757162e+01,1.007810e-01)(-2.562996e+01,1.096050e-01)(-2.368830e+01,1.189025e-01)(-2.174663e+01,1.286082e-01)(-1.980497e+01,1.386262e-01)(-1.786331e+01,1.488259e-01)(-1.592164e+01,1.590389e-01)(-1.397998e+01,1.690577e-01)(-1.203831e+01,1.786391e-01)(-1.009665e+01,1.875111e-01)(-8.154987e+00,1.953858e-01)(-6.213324e+00,2.019774e-01)(-4.271660e+00,2.070237e-01)(-2.329996e+00,2.103096e-01)(-3.883327e-01,2.116879e-01)(1.553331e+00,2.110953e-01)(3.494995e+00,2.085593e-01)(5.436658e+00,2.041947e-01)(7.378322e+00,1.981919e-01)(9.319986e+00,1.907969e-01)(1.126165e+01,1.822890e-01)(1.320331e+01,1.729575e-01)(1.514498e+01,1.630826e-01)(1.708664e+01,1.529203e-01)(1.902830e+01,1.426928e-01)(2.096997e+01,1.325844e-01)(2.291163e+01,1.227404e-01)(2.485329e+01,1.132702e-01)(2.679496e+01,1.042507e-01)(2.873662e+01,9.573164e-02)(3.067829e+01,8.774043e-02)(3.261995e+01,8.028694e-02)(3.456161e+01,7.336762e-02)(3.650328e+01,6.696918e-02)(3.844494e+01,6.107143e-02)(4.038660e+01,5.564963e-02)(4.232827e+01,5.067631e-02)(4.426993e+01,4.612267e-02)(4.621160e+01,4.195952e-02)(4.815326e+01,3.815809e-02)(5.009492e+01,3.469048e-02)(5.203659e+01,3.153005e-02)(5.397825e+01,2.865159e-02)(5.591991e+01,2.603143e-02)(5.786158e+01,2.364754e-02)(5.980324e+01,2.147943e-02)(6.174490e+01,1.950823e-02)(6.368657e+01,1.771652e-02)(6.562823e+01,1.608830e-02)(6.756990e+01,1.460894e-02)(6.951156e+01,1.326501e-02)(7.145322e+01,1.204427e-02)(7.339489e+01,1.093555e-02)(7.533655e+01,9.928631e-03)(7.727821e+01,9.014245e-03)(7.921988e+01,8.183931e-03)(8.116154e+01,7.429994e-03)(8.310320e+01,6.745435e-03)(8.504487e+01,6.123889e-03)(8.698653e+01,5.559571e-03)(8.892820e+01,5.047222e-03)(9.086986e+01,4.582065e-03)(9.281152e+01,4.159758e-03)(9.475319e+01,3.776361e-03)(9.669485e+01,3.428290e-03)(9.863651e+01,3.112293e-03)(9.941318e+01,2.994204e-03)};
\addlegendentry{$E_4=1.8$\,pJ}

\addplot[
    color=blue,
    ]
    coordinates {
(-2.236797e+02,1.330757e-03)(-2.193109e+02,1.465873e-03)(-2.149422e+02,1.614703e-03)(-2.105734e+02,1.778639e-03)(-2.062047e+02,1.959213e-03)(-2.018359e+02,2.158110e-03)(-1.974672e+02,2.377188e-03)(-1.930985e+02,2.618489e-03)(-1.887297e+02,2.884263e-03)(-1.843610e+02,3.176984e-03)(-1.799922e+02,3.499377e-03)(-1.756235e+02,3.854435e-03)(-1.712547e+02,4.245451e-03)(-1.668860e+02,4.676045e-03)(-1.625172e+02,5.150193e-03)(-1.581485e+02,5.672259e-03)(-1.537798e+02,6.247034e-03)(-1.494110e+02,6.879766e-03)(-1.450423e+02,7.576205e-03)(-1.406735e+02,8.342637e-03)(-1.363048e+02,9.185926e-03)(-1.319360e+02,1.011355e-02)(-1.275673e+02,1.113365e-02)(-1.231986e+02,1.225502e-02)(-1.188298e+02,1.348720e-02)(-1.144611e+02,1.484040e-02)(-1.100923e+02,1.632556e-02)(-1.057236e+02,1.795428e-02)(-1.013548e+02,1.973874e-02)(-9.698610e+01,2.169161e-02)(-9.261736e+01,2.382581e-02)(-8.824861e+01,2.615426e-02)(-8.387987e+01,2.868947e-02)(-7.951113e+01,3.144300e-02)(-7.514238e+01,3.442474e-02)(-7.077364e+01,3.764198e-02)(-6.640490e+01,4.109818e-02)(-6.203615e+01,4.479157e-02)(-5.766741e+01,4.871331e-02)(-5.329867e+01,5.284555e-02)(-4.892992e+01,5.715918e-02)(-4.456118e+01,6.161163e-02)(-4.019244e+01,6.614486e-02)(-3.582369e+01,7.068396e-02)(-3.145495e+01,7.513677e-02)(-2.708621e+01,7.939516e-02)(-2.271746e+01,8.333826e-02)(-1.834872e+01,8.683814e-02)(-1.397998e+01,8.976774e-02)(-9.611235e+00,9.201055e-02)(-5.242492e+00,9.347092e-02)(-8.737486e-01,9.408351e-02)(3.494995e+00,9.382016e-02)(7.863738e+00,9.269302e-02)(1.223248e+01,9.075321e-02)(1.660122e+01,8.808528e-02)(2.096997e+01,8.479863e-02)(2.533871e+01,8.101731e-02)(2.970745e+01,7.686999e-02)(3.407620e+01,7.248116e-02)(3.844494e+01,6.796458e-02)(4.281368e+01,6.341904e-02)(4.718243e+01,5.892639e-02)(5.155117e+01,5.455130e-02)(5.591991e+01,5.034230e-02)(6.028866e+01,4.633364e-02)(6.465740e+01,4.254739e-02)(6.902614e+01,3.899575e-02)(7.339489e+01,3.568308e-02)(7.776363e+01,3.260783e-02)(8.213237e+01,2.976408e-02)(8.650112e+01,2.714286e-02)(9.086986e+01,2.473317e-02)(9.523860e+01,2.252281e-02)(9.960735e+01,2.049896e-02)(1.039761e+02,1.864868e-02)(1.083448e+02,1.695915e-02)(1.127136e+02,1.541799e-02)(1.170823e+02,1.401336e-02)(1.214511e+02,1.273404e-02)(1.258198e+02,1.156953e-02)(1.301885e+02,1.051002e-02)(1.345573e+02,9.546416e-03)(1.389260e+02,8.670325e-03)(1.432948e+02,7.874007e-03)(1.476635e+02,7.150357e-03)(1.520323e+02,6.492861e-03)(1.564010e+02,5.895561e-03)(1.607698e+02,5.353011e-03)(1.651385e+02,4.860242e-03)(1.695072e+02,4.412725e-03)(1.738760e+02,4.006331e-03)(1.782447e+02,3.637303e-03)(1.826135e+02,3.302220e-03)(1.869822e+02,2.997971e-03)(1.913510e+02,2.721728e-03)(1.957197e+02,2.470920e-03)(2.000884e+02,2.243210e-03)(2.044572e+02,2.036473e-03)(2.088259e+02,1.848782e-03)(2.131947e+02,1.678382e-03)(2.175634e+02,1.523684e-03)(2.219322e+02,1.383241e-03)(2.236797e+02,1.330757e-03)
};
\addlegendentry{$E_3=0.8$\,pJ}

\addplot[
    color=red,
    dashed,
    ]
    coordinates {
(-8.947186e+02,3.326893e-04)(-8.772436e+02,3.664682e-04)(-8.597687e+02,4.036758e-04)(-8.422937e+02,4.446598e-04)(-8.248187e+02,4.898033e-04)(-8.073437e+02,5.395276e-04)(-7.898688e+02,5.942970e-04)(-7.723938e+02,6.546222e-04)(-7.549188e+02,7.210657e-04)(-7.374439e+02,7.942461e-04)(-7.199689e+02,8.748442e-04)(-7.024939e+02,9.636087e-04)(-6.850189e+02,1.061363e-03)(-6.675440e+02,1.169011e-03)(-6.500690e+02,1.287548e-03)(-6.325940e+02,1.418065e-03)(-6.151190e+02,1.561758e-03)(-5.976441e+02,1.719942e-03)(-5.801691e+02,1.894051e-03)(-5.626941e+02,2.085659e-03)(-5.452192e+02,2.296481e-03)(-5.277442e+02,2.528388e-03)(-5.102692e+02,2.783412e-03)(-4.927942e+02,3.063755e-03)(-4.753193e+02,3.371799e-03)(-4.578443e+02,3.710099e-03)(-4.403693e+02,4.081390e-03)(-4.228943e+02,4.488570e-03)(-4.054194e+02,4.934686e-03)(-3.879444e+02,5.422903e-03)(-3.704694e+02,5.956453e-03)(-3.529945e+02,6.538565e-03)(-3.355195e+02,7.172367e-03)(-3.180445e+02,7.860750e-03)(-3.005695e+02,8.606185e-03)(-2.830946e+02,9.410494e-03)(-2.656196e+02,1.027455e-02)(-2.481446e+02,1.119789e-02)(-2.306696e+02,1.217833e-02)(-2.131947e+02,1.321139e-02)(-1.957197e+02,1.428980e-02)(-1.782447e+02,1.540291e-02)(-1.607698e+02,1.653622e-02)(-1.432948e+02,1.767099e-02)(-1.258198e+02,1.878419e-02)(-1.083448e+02,1.984879e-02)(-9.086986e+01,2.083456e-02)(-7.339489e+01,2.170953e-02)(-5.591991e+01,2.244193e-02)(-3.844494e+01,2.300264e-02)(-2.096997e+01,2.336773e-02)(-3.494995e+00,2.352088e-02)(1.397998e+01,2.345504e-02)(3.145495e+01,2.317325e-02)(4.892992e+01,2.268830e-02)(6.640490e+01,2.202132e-02)(8.387987e+01,2.119966e-02)(1.013548e+02,2.025433e-02)(1.188298e+02,1.921750e-02)(1.363048e+02,1.812029e-02)(1.537798e+02,1.699114e-02)(1.712547e+02,1.585476e-02)(1.887297e+02,1.473160e-02)(2.062047e+02,1.363782e-02)(2.236797e+02,1.258558e-02)(2.411546e+02,1.158341e-02)(2.586296e+02,1.063685e-02)(2.761046e+02,9.748937e-03)(2.935795e+02,8.920771e-03)(3.110545e+02,8.151958e-03)(3.285295e+02,7.441020e-03)(3.460045e+02,6.785714e-03)(3.634794e+02,6.183292e-03)(3.809544e+02,5.630702e-03)(3.984294e+02,5.124741e-03)(4.159044e+02,4.662169e-03)(4.333793e+02,4.239787e-03)(4.508543e+02,3.854498e-03)(4.683293e+02,3.503339e-03)(4.858042e+02,3.183510e-03)(5.032792e+02,2.892381e-03)(5.207542e+02,2.627504e-03)(5.382292e+02,2.386604e-03)(5.557041e+02,2.167581e-03)(5.731791e+02,1.968502e-03)(5.906541e+02,1.787589e-03)(6.081291e+02,1.623215e-03)(6.256040e+02,1.473890e-03)(6.430790e+02,1.338253e-03)(6.605540e+02,1.215061e-03)(6.780289e+02,1.103181e-03)(6.955039e+02,1.001583e-03)(7.129789e+02,9.093257e-04)(7.304539e+02,8.255549e-04)(7.479288e+02,7.494928e-04)(7.654038e+02,6.804321e-04)(7.828788e+02,6.177301e-04)(8.003538e+02,5.608024e-04)(8.178287e+02,5.091183e-04)(8.353037e+02,4.621954e-04)(8.527787e+02,4.195956e-04)(8.702537e+02,3.809211e-04)(8.877286e+02,3.458103e-04)(8.947186e+02,3.326893e-04)
};
\addlegendentry{$E_2=0.2$\,pJ}

\end{semilogyaxis} 
\end{tikzpicture}
    \caption{Used pulses in $4$-EM of truncated solitons ($E_1=0$~pJ not shown).}
    \label{Fig4EMSolitons}
\end{figure}
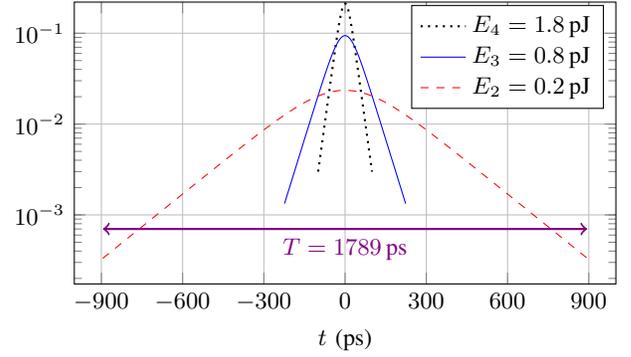

Based on \eqref{eq:Tmax} and Fig. \ref{FigTsolitons}, $T=T_{p_2}$ for $M\geq4$ since $E_2$ falls in the region with no broadening. Using $E_M=1.8$~pJ (i.e. $p_M(t)$ is the truncated soliton with bandwidth $W_{max}$) maximizes the transmission rate for $M\geq4$ where the resulting $T=T_{p_2}$ is equal to the right hand side of \eqref{eq:ThProofTp2}. Hence, the transmission rates achieve the upper bound of Theorem \ref{Th:UpperBound} with equality $R=5\frac{\log_2M}{(M-1)^2}$~Gbit/s for $M \geq 4$ (for the considered simulation parameters). The main weakness of solitons shows up here: OOK outperforms $M$-EM systems with $M\geq4$ and the rates decrease if $M$ increases.  It is possible to look for the optimal energy levels (instead of the ones in (\ref{Eq:EnergiesEM})) to improve the rates as done in  \cite{chen2020capacity}, where the optimal achievable rate depends on the noise model. In this paper we do not search for the optimal energy levels since we consider noiseless fibers and because our main goal is to introduce the MTB waveforms.

\begin{table}[t!]
    \centering
    \caption{Transmission rates (Gbit/s) for fiber length $L=80$Km and at an isolation level $(1-10^{-4})$. We show the percentage increase with respect to OOK of truncated solitons.}
    \begin{tabular}{|c|c|c|}
    \hline   & Truncated solitons  & MTB pulses    \\
    \hline    OOK, lossless fibers&  $4.18$  &  $4.68$ ($12\%$ increase)\\
    \hline    OOK, lossy fibers &  $2.96$ &  $3.96$ ($33.8\%$ increase)\\
    \hline    $4$-EM, lossless fibers& $1.12$ & $7.12$ ($70\%$ increase)\\
    \hline    $4$-EM, lossy  fibers& $1.12$   & $7.09$ (164$\%$ increase)\\
    \hline    \begin{tabular}{@{}c@{}}$M$-EM, $M \geq 4$ \\ Lossless \& lossy fibers\end{tabular}  & $ 5\frac{\log_2M}{(M-1)^2}$ & $\approx 3.5\log_2M$\\
    \hline
    \end{tabular}
    \label{TableRates}
\end{table}

For sample values of the energy level $E$, we compute  the time-bandwidth product of truncated solitons (based on   \eqref{eq:Tmax}, \eqref{eq:WeqMax}, and \eqref{eq:TimBandProd}) as
\begin{equation*}
    c=\max\{T_s, T_{\tilde{s}}\}\cdot \max\{W_s, W_{\tilde{s}}\}.
\end{equation*}
The obtained numerical results are shown in Fig. \ref{FigTBproduct}. 
Note that for $E\leq 0.9$~pJ, the time-bandwidth product for truncated solitons is almost equal to the one of nontruncated solitons ($c_s=9.94$ computed using (\ref{cSoliton})), since the truncated solitons do not broaden for this range of $E$ (as shown in Fig. \ref{FigTsolitons}).
Using the results of Fig. \ref{FigTBproduct}, we compute the SE in \eqref{eq:SE} for OOK (shown in Fig. \ref{FigTsolitons}).  The obtained SEs are shown in Table \ref{TableSE}. Note that the spectral efficiency increases as $\epsilon$ increases, where $0.33$~bit/s/Hz was reported in \cite{yousefi2014Part3} for OOK of solitons when using $\epsilon=10^{-2}$. The used $\epsilon=10^{-4}$ here is considered ``practically reasonable'' in \cite{span2019time}, where the approximated time-bandwidth product of fundamental solitons in \cite{span2019time} is similar to the one obtained here by (\ref{cSoliton}) and yields a similar spectral efficiency around $0.1$~bit/s/Hz for OOK of solitons.

For $M$-EM of isolated pulses, the time-bandwidth product depends on $T$  and $W$ that are defined in (\ref{eq:Tmax}) and (\ref{eq:WeqMax}). For the $4$-EM of solitons shown in Fig. \ref{FigTsolitons}, $W=50$~GHz which occurs at $E=1.8$~pJ, and  $T=1789$~ps which is the duration of the soliton at $0.2$~pJ, then the resulting time-bandwidth product is $\approx89.5$ and the corresponding spectral efficiency is $\log_2(4)/89.5\approx 0.022$~bit/s/Hz. The spectral efficiency for $M\geq4$ can be obtained form Table \ref{TableRates} by dividing the transmission rate by $50$~GHz (the bandwidth at $1.8$~pJ). Similarly to the transmission rates of truncated solitons, the spectral efficiency decreases when $M$ increases.

\begin{figure}[t!]
\centering
\begin{tikzpicture}[font=\small,
    Gclabel/.style={text=black},]
\begin{axis}[
    ylabel={Time-Bandwidth product},
    xlabel={$E_s$ (pJ)},
    xmin=0, xmax=1.8,
    ymin=7.8, ymax=20,
    xtick={0,0.3,0.6,0.9,1.2,1.5,1.8},
    ytick={8,10,12,14,16,18,20},
    legend style={at={(0.35,0.98)},anchor=north},
    ymajorgrids=true,
    xmajorgrids =true,
    legend cell align={left},    
    scale only axis,
    width=7cm,
    height=4.5cm,
]


\addplot[color=black ]coordinates{(0.01, 9.94)(0.6, 9.94)(1.2, 9.94)(1.8, 9.94)};\addlegendentry{Solitons ($c_s=9.94$)}

\addplot[color=red,
    dashed,dash pattern=on 4pt off 4pt on 4pt off 4pt,
    ]
coordinates{
(1.000000e-04,1.000967e+01)(3.000000e-01,9.976420e+00)(6.000000e-01,1.0033668e+01)(9.000000e-01,9.996654e+00)(1.200000e+00,1.073821e+01)(1.500000e+00,1.132986e+01)(1.800000e+00,1.208904e+01)

};\addlegendentry{Truncated Solitons, Lossless}
\addplot[color=blue, 
    dash pattern=on 5pt off 1pt on 1pt off 1pt,
    ]
    coordinates{(1.000000e-04,1.000967e+01)(3.000000e-01,9.965335e+00)(6.000000e-01,9.943165e+00)(9.000000e-01,9.996654e+00)(1.200000e+00,1.124126e+01)(1.500000e+00,1.449000e+01)(1.800000e+00,1.919399e+01)
};\addlegendentry{Truncated Solitons, Lossy}

\definecolor{clr1}{RGB}{50,230,50}

\addplot[color=clr1, mark=square, ]coordinates{(0.01, 8.37)(0.3, 8.37)(0.6, 8.37)(0.9, 8.37)(1.2, 8.37)(1.5, 8.37)(1.8, 8.37)};\addlegendentry{MTB, Dispersion-only}

\addplot[color=red,
    dashed,
    dash pattern=on 4pt off 4pt on 4pt off 4pt,
    mark=star,
    mark options={solid},
    ]
    coordinates {
(1.000000e-04,8.372861e+00)(3.000000e-01,8.703292e+00)(6.000000e-01,8.98)(0.9,9.252926e+00)(1.200000e+00,9.597602e+00)(1.500000e+00,1.017378e+01)(1.800000e+00,1.068718e+01) };
\addlegendentry{MTB, Lossless}

\addplot[
   color=blue,dash pattern=on 5pt off 1pt on 1pt off 1pt, 
    mark=o, 
    mark options={solid},
    ]
    coordinates {
(1.000000e-04,8.372861e+00)(3.000000e-01,8.382047e+00)(6.000000e-01,8.422805e+00)(0.9,8.347031e+00)(1.200000e+00,8.344789e+00)(1.500000e+00,8.975827e+00)(1.800000e+00,1.158314e+01)
};
\addlegendentry{MTB, Lossy}

\draw[color=orange, rotate around={0:(0.45,10)}](0.45,10) ellipse  (0.03 and 0.6);
\draw[orange, ->] (0.45,11.5)--  (0.45,10.7) node[pos=-0.8] {Solitons};

\draw[color=orange, rotate around={0:(0.75,8.75)}](0.75,8.75) ellipse  (0.03 and 0.75);
\draw[orange, ->] (0.75,11.5)--  (0.75,9.6) node[pos=-0.3] {MTB};

\end{axis}    
\end{tikzpicture}
    \caption{Time-bandwidth product of MTB pulses and solitons for $L=80$~Km and $\epsilon=10^{-4}$.}
    \label{FigTBproduct}
\end{figure}
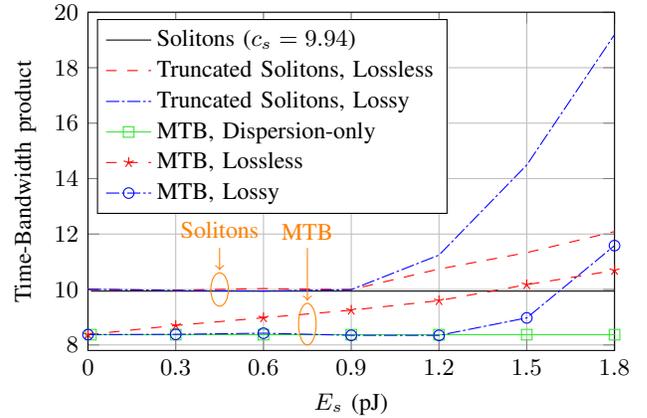

\begin{table}[t!]
    \centering
    \caption{Spectral efficiency (bit/s/Hz) for fiber length $L=80$~Km and at an isolation level $(1-10^{-4})$.}
    \begin{tabular}{|c|c|c|}
    \hline   & Truncated solitons  & MTB pulses    \\
    \hline    OOK, lossless fibers&  $0.084$  &  $0.094$\\
    \hline     OOK, lossy fibers &  $0.089$ &  $0.111$\\
    \hline    \begin{tabular}{@{}c@{}}$4$-EM,  \\ Lossless \& lossy fibers\end{tabular}& $0.022$   & $0.14$\\
    \hline    \begin{tabular}{@{}c@{}}$M$-EM, $M \geq 4$ \\ Lossless \& lossy fibers\end{tabular}  & $ 0.1\frac{\log_2M}{(M-1)^2}$ & $\approx 0.07\log_2M$\\
    \hline
    \end{tabular}
    \label{TableSE}
\end{table}

\section{Energy Modulation of Waveforms with Minimum Time Broadening}
\label{Sec:MTBWaveforms}

We consider maximizing the transmission rate in (\ref{eq:rates}) by finding the optimal pulses $\left\{p_m(t)\right\}_{m\in\mathcal{M}}$ that allows to use the minimum possible $T$ (defined in \eqref{eq:Tmax}). Minimizing the right hand side of (\ref{eq:Tmax}) can be done by independently finding the optimal pulses that minimize the maximum of each pair $\{T_{p_m}, T_{\Tilde{p}_m}\}_{m=2}^{M}$.

We consider the general problem of finding the optimal pulse $p(t)$ that has an energy $E$ and achieves the minimum possible value of $\max\{T_p, T_{\Tilde{p}}\}$. To solve this problem, we first consider the problem of finding the pulse $p(t)$ that achieves the minimum $T_{\Tilde{p}}$ for a given value of $T_p$:
\begin{equation}
   \begin{aligned}
    p^*(t) &=  \,\,\,\, \underset{p(t)}{\text{argmin}} && T_{\Tilde{p}}\\
   & \text{subject to} && p(t)=0 \text{ for } t\notin [-T_p/2, \, T_p/2] \\
   & && \int_{-T_{p}/2}^{T_{p}/2} |p(t)|^2dt = E\\
   & && \int_{-W_{max}}^{W_{max}} \left|P(f)\right|^2df \geq (1-\epsilon)E
   \end{aligned},
   \label{eq:OptProb}
\end{equation}
where the third constraint is used to satisfy (\ref{eq:BandConstr}).
To find the optimal solution of (\ref{eq:OptProb}), we use a numerical method that is similar to the one proposed in \cite{jaffal2022time}; we write $p(t)$ as a linear combination of normalized truncated PSWFs (see (9) in \cite{jaffal2022time} and the appendix therein) and we find the optimal  combinations of normalized truncated PSWFs that minimize the objective function in (\ref{eq:OptProb}). We use the software in \cite{adelman2014software} to generate the PSWFs, and we use the MATLAB function \textit{`fmincon'} to find the optimal solution of (\ref{eq:OptProb}).

For chosen sample values of $T_{p}$, we search for the corresponding optimal pulses $p^*(t)$ that yield $T_{\Tilde{p}^*}$ at the receiver. 
Using higher values of $T_p$ implies a more relaxed constraint on the time-limitation of $p(t)$, and hence, the corresponding $T_{\Tilde{p}^*}$ does not increase as we increase $T_p$. Therefore the minimum value of $\max\{T_p, T_{\Tilde{p}}\}$ is achieved when $T_{\Tilde{p}^*}=T_p$. We refer to the corresponding optimal pulse as the MTB pulse since it achieves the minimum time broadening.



In the following, we present the numerical results for three models of the channel: dispersion-only channel, lossless channel, and lossy channel. Although the lossless and lossy channels are more practically relevant, we present the results for the dispersion-only channel too since it approximates well the fiber model in the low power regime. 

\subsection{Dispersion-only channel}
\label{Sec:NumResDO}

We consider the following values of the parameters in (\ref{NLSEq}): $\beta_2=-21.7$~$\text{ps}^2$/Km, $\gamma=0$~$\text{W}^{-1}\text{Km}^{-1}$, and $\alpha=0$~dB/Km. As the dispersion-only channel is linear, the time broadening of a pulse is not affected by its energy. Consequently, the solutions of (\ref{eq:OptProb}) for different $E$ are scaled versions of each other and they yield the same $T_{\tilde{p}^*}$. In this section we present the results of (\ref{eq:OptProb}) for $E=1$~pJ. 

\begin{figure}[t!]
\centering
\begin{tikzpicture}[font=\small,
    Gclabel/.style={text=black},]
\begin{axis}[
    xlabel={$T_p$ (ps)},
    ylabel={$T_{\tilde{p}}$ (ps)},
    xmin=200, xmax=400,
    ymin=265, ymax=330,
    xtick={200,240, 280, 320, 360, 400},
    ytick={280,320},
    legend style={at={(0.84,0.99)},anchor=north},
    ymajorgrids=true,
    xmajorgrids =true,
    scale only axis,
    width=7cm,
    height=3.8cm,
]

\addplot[
    color=red,
    mark=square,
    ]
    coordinates {
(200.000,321.996)(240.000,297.876)(280.000,286.869)(320.000,279.141)(360.000,276.800)(400.000,275.551)
    };
\addlegendentry{$T_{\tilde{p}^*}$}

\addplot[
    color=black,
    dashed,
    ]
    coordinates {
(260,260)(330,330)
    };
\addlegendentry{$T_{\tilde{p}}=T_p$}

\addplot[
    color=blue,
    mark=*,
    only marks,
    ]
    coordinates {
(285.5,285.5)
    };
\draw[blue, ->] (285.5,274) -- (285.5,283.75) node[pos=-0.5,xshift=40] {Minimum of $\max\{T_{p},T_{\tilde{p}^*}\}$};

\draw[black!75, rotate around={-38:(241,304)}](241,304) ellipse  (53 and 8);
\draw[black, ->] (260,322) -- (248,312) node[pos=-0.4] {$\max\{T_{p},T_{\tilde{p}^*}\}=T_{\tilde{p}^*}$};

\draw[black!75, rotate around={-8.5:(343.5,280.5)}](343.5,280.5) ellipse  (57 and 6);
\draw[black, ->] (345,295) -- (341,287.5) node[pos=-0.4] {$\max\{T_{p},T_{\tilde{p}^*}\}=T_p$};

\end{axis}    
\end{tikzpicture}
\caption{Obtained $T_{\tilde{p}^*}$ for sample values of $T_p$, dispersion-only channel for $E=1$~pJ, $L=80$~Km, and $\epsilon=10^{-4}$.}
    \label{FigODc10to20}
\end{figure}
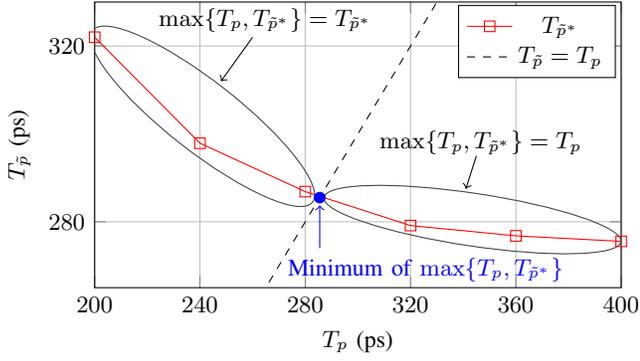

We found the optimal solution $p^*(t)$ of (\ref{eq:OptProb}) for $T_p\in\{200, 240, 280, 320, 360, 400 \}$~ps and in Fig. \ref{FigODc10to20} we show the obtained $T_{\tilde{p}^*}$. The straight line in Fig. \ref{FigODc10to20} creates two regions, in the left one $\max\{T_p, T_{\Tilde{p}^*}\}=T_{\Tilde{p}^*}$ and in the right one $\max\{T_p, T_{\Tilde{p}^*}\}=T_{p}$.
At the intersection, the minimum value of $\max\{T_p, T_{\Tilde{p}^*}\}$ is achieved where $T_{\Tilde{p}^*}=T_p$, hence the modulation interval (defined in (\ref{eq:Tmax})) is $T\approx285.5$~ps and the transmission rate $R$ in \eqref{eq:rates} is $3.5\log_2M$~Gbit/s.
In Fig. \ref{FigODpulse} we show the propagation of the obtained pulse $p^*(t)$ that achieves the minimum time broadening for $T_p=285.5$~ps. We refer to this pulse as the MTB pulse. 
Using the bandwidth definition in (\ref{eq:Banddef}), we found that the effective bandwidth of $p^*(t)$ is equal to $29.3$~GHz, which is also equal to the effective bandwidth of $\tilde{p}^*(t)$ since the dispersion-only channel models an all-pass filter. Then the spectral efficiency defined in (\ref{eq:SE}) is equal to $0.12\log_2M$~bit/s/Hz for the dispersion only channel. 


The obtained transmission rate depends mainly on the length of the fiber and the chosen isolation level. Therefore, higher transmission rates can be obtained for lower distances or for a lower level of isolation (i.e., higher $\epsilon$). For example, the MTB pulse for $L=40$~Km is found to have $T_{p^*}=T_{\tilde{p}^*}=202$~ps which yields a transmission rate of $5\log_2{M}$~Gbit/s, i.e., an increase of $66.7$\%. However the improvement in the spectral efficiency is lower (an increase of $2.8$\%) since the effective bandwidth of the MTB pulse is $40.3$~GHz (compared to $29.3$~GHz for $L=80$~Km).

\begin{figure}[t!]
\centering
\input{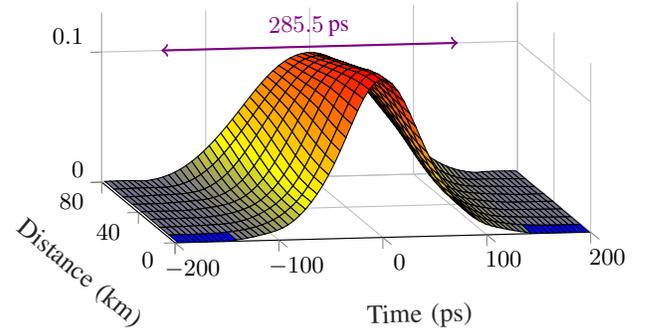}
\caption{Propagation of the MTB pulse, dispersion-only channel for $E=1$~pJ and $L=80$~Km.}
    \label{FigODpulse}
\end{figure}

Since the dispersion-only channel is a linear channel, the received signal is $|q(t,L)| = \sum_{k\in\mathbb{Z}} \Tilde{p}_{m_k}(t-kT)$ when $q(t,0) = \sum_{k\in\mathbb{Z}} p_{m_k}(t-kT)$ is transmitted even if the consecutive pulses are not isolated, i.e., if $T$ does not satisfy \eqref{eq:TxIsolation} and \eqref{eq:RxIsolation}. In this case, decreasing $T$ requires an equalizer at the receiver to suppress the resulting intersymbol interference. The complexity of the equalizer depends on the number of interfering symbols. For example, the complexity of the Viterbi equalizer increases exponentially with the number of interfering symbols \cite{jaffal2022time}.  In this context, the MTB pulses are desirable as they result in a low number of interfering symbols.

\subsection{Lossless and Lossy Channels}
\label{SectCompare}



In this section we use the simulation parameters that we used in Sec. \ref{Sec:Solitons} ($L=80$~Km, $\epsilon=10^{-4}$, $W_{max}=50$~GHz, $\beta_2=-21.7$~$\text{ps}^2$/Km, $\gamma=1.2$~$\text{W}^{-1}\text{Km}^{-1}$, $\alpha=0$~dB/Km for the lossless fiber, and $\alpha=0.2$~dB/Km for the lossy fiber).  
We obtained the MTB pulses for sample values of the energy level $E$ that are allowed by solitons, i.e., for $E\leq 1.8$~pJ (for the considered simulation parameters).
For a given energy level $E$, we numerically solve (\ref{eq:OptProb}) for sample values of $T_p$ and we find the MTB pulse that minimizes $\max\{T_p, T_{\Tilde{p}}\}$ (as done in Sec. \ref{Sec:NumResDO} and Fig. \ref{FigODc10to20}).
We show the obtained duration of MTB pulses in Fig. \ref{FigTMTB}, where the results for the dispersion-only channel do not change with $E$ since the dispersion-only channel is linear. The results for lossless and lossy fibers converge to the one of the dispersion-only fibers as $E\to 0$. This is expected since the Kerr effect has lower impact in lower power regimes, where the NLS equation can be approximated by neglecting the nonlinear term (i.e., using $\gamma=0$). For the lossless fibers, the duration of the MTB pulses decreases as $E$ increases, where the minimum duration is $213.9$~ps (achieved at $E=1.8$~pJ). For lossy fibers, the duration of the MTB pulses increases as $E$ increases for $E\leq1.5$~pJ, reaches a minimum of $252.7$~ps at $E=1.5$~pJ, and then increases as $E$ increases beyond $1.5$~pJ. 

The transmission rates for OOK are computed based on the minimum duration (as we did for truncated solitons in Sec. \ref{Sec:Solitons}), and we include the corresponding transmission rates in the second column of Table \ref{TableRates}. The obtained results show improvements (compared to truncated solitons) of $33.8$\% and $12$\% for OOK over lossy and lossless fibers respectively. For $4$-EM, we show in Fig. \ref{FigTMTB} the energy levels that maximize the transmission rates given by (\ref{eq:rates}) while satisfying (\ref{Eq:EnergiesEM}) and (\ref{eq:Tmax}), where the values of $T$ are $282$~ps and $281$~ps for lossy and lossless fibers respectively. The transmission rates are then computed and shown in Table \ref{TableRates}. Note that we compute the percentage increase with respect to OOK of truncated  solitons since OOK is the best option when using truncated solitons. We also derive an approximation for the transmission rates of $M$-EM of MTB pulses for $M\geq4$; The modulation interval $T$ can be approximated by the duration of the MTB pulse of the dispersion-only fiber $T\approx285.5$~ps, and hence the transmission rates are approximated by $3.5\log_2{M}$~Gbit/s. Therefore, increasing $M$ increases the transmission rates of EM of isolated MTB pulses, which is the main advantage of MTB pulses over the truncated solitons. Comparing Fig. \ref{FigTsolitons} and Fig. \ref{FigTMTB}, we notice that the obtained MTB pulses achieve lower values of the minimum of $\max\{T_p, T_{\Tilde{p}}\}$ compared to truncated solitons, and hence they outperform truncated solitons in the considered EM system. In Fig. \ref{FigMTB4EM} and  Fig. \ref{FigRxMTB4EM} we show the transmitted and received MTB pulses in the  $4$-EM system over lossless fiber, which shows that it is possible to use low energy pulses at the cost of a slight increase in $T$. 

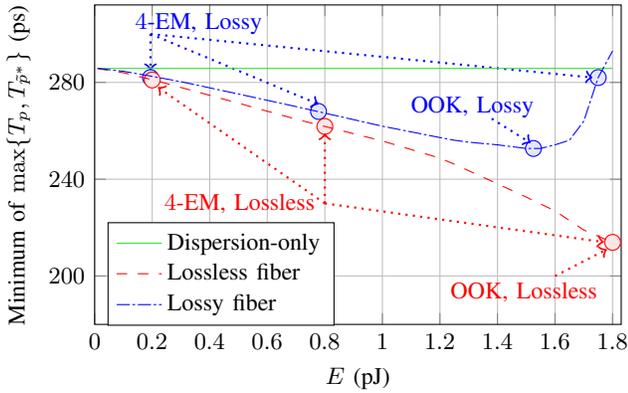
\begin{figure}[t!]
\centering
\begin{tikzpicture}[font=\small,
    Gclabel/.style={text=black},]
\begin{axis}[
    xlabel={$E$ (pJ)},
    ylabel={Minimum of $\max\{T_{p},T_{\tilde{p}^*}\}$ (ps)},
    xmin=0, xmax=1.83,
    ymin=180, ymax=312,
    xtick={0,0.2,0.4,0.6,0.8, 1, 1.2,1.4, 1.6,1.8},
    ytick={200,240,280,320},
    legend style={at={(0.24,0.32)},anchor=north},
    ymajorgrids=true,
    xmajorgrids =true,
    legend cell align={left},    
    scale only axis,
    width=7cm,
    height=4.25cm,
]

\filldraw[color=blue, fill=blue!10](1.75,282) ellipse  (3pt and 3pt);
\draw[blue, ->, dotted, thick] (0.194,300) -- (1.73,282);

\filldraw[color=blue, fill=blue!10](0.778,268) ellipse  (3pt and 3pt);
\draw[blue, ->, dotted, thick] (0.194,300) -- (0.778,271)node[pos=-0.15, right] {$4$-EM, Lossy};

\filldraw[color=blue, fill=blue!10](0.194,282) ellipse  (3pt and 3pt);
\draw[blue, ->, dotted, thick] (0.194,300) -- (0.194,285);

\filldraw[color=red, fill=red!10](1.8,213.9) ellipse  (3pt and 3pt);
\draw[red, ->, dotted, thick] (0.8,230) -- (1.77,214)node[pos=-0.00, left] {$4$-EM, Lossless};

\filldraw[color=red, fill=red!10](0.8,261.8) ellipse  (3pt and 3pt);
\draw[red, ->, dotted, thick] (0.8,230) -- (0.8,259);

\filldraw[color=red, fill=red!10](0.2,281) ellipse  (3pt and 3pt);
\draw[red, ->, dotted, thick] (0.8,230) -- (0.22,278);

\draw[red, ->, dotted, thick] (1.6,200) -- (1.78,211) node[pos=-0.6] {OOK, Lossless};

\filldraw[color=blue, fill=blue!10](1.525,252.7) ellipse  (3pt and 3pt);
\draw[blue, ->, dotted, thick] (1.38,265) -- (1.515,255) node[pos=-0.5] {OOK, Lossy};

\definecolor{clr3}{RGB}{50,230,50}

\addplot[color=clr3, ]coordinates{(0.01, 285.8)(1.8, 285.8)};\addlegendentry{Dispersion-only}

\addplot[
    color=red,
    dashed,
    dash pattern=on 4pt off 4pt on 4pt off 4pt,
    ]
    coordinates {
(0.0100,285.7)(0.1000,284.0)(0.2000,281.1)(0.3000,277.9)(0.5000,271.4)(0.7000,264.8)(1.0000,255.8)(1.2000,248.7)(1.2500,246.2)(1.3000,243.5)(1.4000,238.1)(1.5000,232.6)(1.6000,226.9)(1.7000,219.0)(1.8000,213.9)};
\addlegendentry{Lossless  fiber}

\addplot[
    color=blue,dotted, dash pattern=on 5pt off 1pt on 1pt off 1pt,
    ]
    coordinates {
    (0.0100,285.8)(0.1000,284.6)(0.2000,282.5)(0.3000,280.3)(0.5000,275.2)(0.7000,270.0)(1.0000,261.9)(1.2000,257.2)(1.2500,256.1)(1.3000,255.3)(1.4000,254.2)(1.5000,252.7)(1.5500,252.7)(1.6000,254.2)(1.6500,256.1)(1.7000,264.0)(1.7500,282.0)(1.8000,293.0)};
\addlegendentry{Lossy fiber}

\end{axis}    
\end{tikzpicture}
    \caption{Achieved duration of MTB pulses for $L=80$~Km and $\epsilon=10^{-4}$.}
    \label{FigTMTB}
\end{figure}


We compute the time-bandwidth product of the MTB pulses using
\begin{equation*}
   c= \max\{T_{p^*}, T_{\tilde{p}^*}\}\cdot \max\{W_{p^*}, W_{\tilde{p}^*}\},
\end{equation*}
and we include them in Fig. \ref{FigTBproduct}. The obtained results show that the MTB pulses have lower time-bandwidth product compared to truncated solitons. Similarly to Sec. \ref{Sec:Solitons}, we compute the spectral efficiency of OOK and $4$-EM of MTB pulses and we include them in Table \ref{TableSE}. The approximation of the spectral efficiency for $M\geq4$ is obtained by dividing the approximation of the transmission rates (in Table \ref{TableRates}) by the maximum allowable bandwidth ($50$~GHz). As shown in Tables \ref{TableRates} and \ref{TableSE}, the proposed MTB waveforms outperform the truncated solitons in EM of isolated pulses in terms of transmission rate and spectral efficiency. Moreover, the transmission rate and spectral efficiency of MTB pulses increase as $M$ increases, which is not the case for fundamental solitons.

\begin{figure}[t!]
\centering
\begin{tikzpicture}[font=\small,
    Gclabel/.style={text=black},]
\begin{axis}[
    xlabel={$t$ (ps)},
    xmin=-150, xmax=150,
    ymin=0.0, ymax=0.22,
    ytick={0.1, 0.2},
    xtick={-150,-100,-50,0,50,100,150},
    legend style={at={(0.8,0.99)},anchor=north},
    ymajorgrids=true,
    xmajorgrids =true,
    scale only axis,
    width=7cm,
    height=2.7cm,
]

\draw[violet, <->,  thick] (-140.5,0.095) -- (140.5,0.095)node[pos=0.25, above] {$T=281$\,ps};

\addplot[
    color=black,
    thick,
    dotted,
    ]
    coordinates {(-1.038876e+02,1.150890e-03)(-1.018386e+02,1.170321e-03)(-9.978952e+01,1.617695e-03)(-9.774045e+01,2.246079e-03)(-9.569139e+01,2.898327e-03)(-9.364232e+01,3.495164e-03)(-9.159325e+01,4.017796e-03)(-8.954419e+01,4.488151e-03)(-8.749512e+01,4.949595e-03)(-8.544606e+01,5.450399e-03)(-8.339699e+01,6.031500e-03)(-8.134792e+01,6.719280e-03)(-7.929886e+01,7.523316e-03)(-7.724979e+01,8.438354e-03)(-7.520072e+01,9.449331e-03)(-7.315166e+01,1.053800e-02)(-7.110259e+01,1.168968e-02)(-6.905353e+01,1.289891e-02)(-6.700446e+01,1.417304e-02)(-6.495539e+01,1.553338e-02)(-6.290633e+01,1.701385e-02)(-6.085726e+01,1.865759e-02)(-5.880820e+01,2.051232e-02)(-5.675913e+01,2.262539e-02)(-5.471006e+01,2.503936e-02)(-5.266100e+01,2.778910e-02)(-5.061193e+01,3.090073e-02)(-4.856287e+01,3.439267e-02)(-4.651380e+01,3.827860e-02)(-4.446473e+01,4.257165e-02)(-4.241567e+01,4.728905e-02)(-4.036660e+01,5.245616e-02)(-3.831754e+01,5.810895e-02)(-3.626847e+01,6.429409e-02)(-3.421940e+01,7.106606e-02)(-3.217034e+01,7.848114e-02)(-3.012127e+01,8.658868e-02)(-2.807221e+01,9.542027e-02)(-2.602314e+01,1.049780e-01)(-2.397407e+01,1.152232e-01)(-2.192501e+01,1.260671e-01)(-1.987594e+01,1.373646e-01)(-1.782687e+01,1.489129e-01)(-1.577781e+01,1.604551e-01)(-1.372874e+01,1.716886e-01)(-1.167968e+01,1.822797e-01)(-9.630611e+00,1.918817e-01)(-7.581544e+00,2.001551e-01)(-5.532478e+00,2.067902e-01)(-3.483412e+00,2.115276e-01)(-1.434346e+00,2.141764e-01)(6.147198e-01,2.146280e-01)(2.663786e+00,2.128637e-01)(4.712852e+00,2.089561e-01)(6.761918e+00,2.030642e-01)(8.810984e+00,1.954214e-01)(1.086005e+01,1.863194e-01)(1.290912e+01,1.760882e-01)(1.495818e+01,1.650742e-01)(1.700725e+01,1.536189e-01)(1.905631e+01,1.420392e-01)(2.110538e+01,1.306121e-01)(2.315445e+01,1.195632e-01)(2.520351e+01,1.090610e-01)(2.725258e+01,9.921667e-02)(2.930164e+01,9.008828e-02)(3.135071e+01,8.168901e-02)(3.339978e+01,7.399782e-02)(3.544884e+01,6.697120e-02)(3.749791e+01,6.055450e-02)(3.954698e+01,5.469161e-02)(4.159604e+01,4.933217e-02)(4.364511e+01,4.443570e-02)(4.569417e+01,3.997272e-02)(4.774324e+01,3.592310e-02)(4.979231e+01,3.227262e-02)(5.184137e+01,2.900836e-02)(5.389044e+01,2.611427e-02)(5.593950e+01,2.356760e-02)(5.798857e+01,2.133704e-02)(6.003764e+01,1.938289e-02)(6.208670e+01,1.765928e-02)(6.413577e+01,1.611806e-02)(6.618483e+01,1.471362e-02)(6.823390e+01,1.340780e-02)(7.028297e+01,1.217388e-02)(7.233203e+01,1.099884e-02)(7.438110e+01,9.883212e-03)(7.643016e+01,8.838528e-03)(7.847923e+01,7.882361e-03)(8.052830e+01,7.031905e-03)(8.257736e+01,6.297128e-03)(8.462643e+01,5.674921e-03)(8.667549e+01,5.145718e-03)(8.872456e+01,4.673913e-03)(9.077363e+01,4.213017e-03)(9.282269e+01,3.715893e-03)(9.487176e+01,3.149690e-03)(9.692083e+01,2.514242e-03)(9.896989e+01,1.861955e-03)(1.010190e+02,1.316523e-03)(1.030680e+02,1.087446e-03)
};
\addlegendentry{$E_4=1.8$~pJ}

\addplot[
    color=blue,
    ]
    coordinates {
(-1.298596e+02,9.798694e-04)(-1.272982e+02,1.230492e-03)(-1.247369e+02,1.497140e-03)(-1.221756e+02,1.774403e-03)(-1.196142e+02,2.057770e-03)(-1.170529e+02,2.344338e-03)(-1.144916e+02,2.633440e-03)(-1.119302e+02,2.927098e-03)(-1.093689e+02,3.230272e-03)(-1.068076e+02,3.550855e-03)(-1.042462e+02,3.899429e-03)(-1.016849e+02,4.288762e-03)(-9.912357e+01,4.733130e-03)(-9.656224e+01,5.247488e-03)(-9.400091e+01,5.846599e-03)(-9.143957e+01,6.544189e-03)(-8.887824e+01,7.352237e-03)(-8.631691e+01,8.280472e-03)(-8.375558e+01,9.336134e-03)(-8.119424e+01,1.052407e-02)(-7.863291e+01,1.184711e-02)(-7.607158e+01,1.330683e-02)(-7.351025e+01,1.490441e-02)(-7.094891e+01,1.664180e-02)(-6.838758e+01,1.852279e-02)(-6.582625e+01,2.055407e-02)(-6.326491e+01,2.274598e-02)(-6.070358e+01,2.511299e-02)(-5.814225e+01,2.767360e-02)(-5.558092e+01,3.044983e-02)(-5.301958e+01,3.346601e-02)(-5.045825e+01,3.674710e-02)(-4.789692e+01,4.031649e-02)(-4.533559e+01,4.419344e-02)(-4.277425e+01,4.839029e-02)(-4.021292e+01,5.290964e-02)(-3.765159e+01,5.774182e-02)(-3.509026e+01,6.286265e-02)(-3.252892e+01,6.823191e-02)(-2.996759e+01,7.379258e-02)(-2.740626e+01,7.947102e-02)(-2.484493e+01,8.517816e-02)(-2.228359e+01,9.081170e-02)(-1.972226e+01,9.625928e-02)(-1.716093e+01,1.014024e-01)(-1.459960e+01,1.061212e-01)(-1.203826e+01,1.102991e-01)(-9.476931e+00,1.138282e-01)(-6.915598e+00,1.166142e-01)(-4.354265e+00,1.185803e-01)(-1.792933e+00,1.196716e-01)(7.683998e-01,1.198571e-01)(3.329732e+00,1.191315e-01)(5.891065e+00,1.175154e-01)(8.452398e+00,1.150544e-01)(1.101373e+01,1.118167e-01)(1.357506e+01,1.078900e-01)(1.613640e+01,1.033773e-01)(1.869773e+01,9.839221e-02)(2.125906e+01,9.305382e-02)(2.382039e+01,8.748183e-02)(2.638173e+01,8.179160e-02)(2.894306e+01,7.608995e-02)(3.150439e+01,7.047164e-02)(3.406572e+01,6.501680e-02)(3.662706e+01,5.978925e-02)(3.918839e+01,5.483601e-02)(4.174972e+01,5.018762e-02)(4.431105e+01,4.585943e-02)(4.687239e+01,4.185353e-02)(4.943372e+01,3.816118e-02)(5.199505e+01,3.476562e-02)(5.455638e+01,3.164480e-02)(5.711772e+01,2.877410e-02)(5.967905e+01,2.612871e-02)(6.224038e+01,2.368550e-02)(6.480171e+01,2.142444e-02)(6.736305e+01,1.932938e-02)(6.992438e+01,1.738833e-02)(7.248571e+01,1.559320e-02)(7.504704e+01,1.393914e-02)(7.760838e+01,1.242358e-02)(8.016971e+01,1.104515e-02)(8.273104e+01,9.802580e-03)(8.529238e+01,8.693641e-03)(8.785371e+01,7.714371e-03)(9.041504e+01,6.858532e-03)(9.297637e+01,6.117394e-03)(9.553771e+01,5.479847e-03)(9.809904e+01,4.932798e-03)(1.006604e+02,4.461809e-03)(1.032217e+02,4.051898e-03)(1.057830e+02,3.688408e-03)(1.083444e+02,3.357867e-03)(1.109057e+02,3.048733e-03)(1.134670e+02,2.751976e-03)(1.160284e+02,2.461430e-03)(1.185897e+02,2.173896e-03)(1.211510e+02,1.888993e-03)(1.237124e+02,1.608781e-03)(1.262737e+02,1.337203e-03)(1.288350e+02,1.079391e-03)
};
\addlegendentry{$E_3=0.8$~pJ}

\addplot[
    color=red,
    dashed,
    ]
    coordinates {
(-1.366943e+02,3.541698e-04)(-1.339981e+02,4.968488e-04)(-1.313020e+02,6.518936e-04)(-1.286059e+02,8.163729e-04)(-1.259097e+02,9.877180e-04)(-1.232136e+02,1.164074e-03)(-1.205174e+02,1.344613e-03)(-1.178213e+02,1.529774e-03)(-1.151252e+02,1.721422e-03)(-1.124290e+02,1.922886e-03)(-1.097329e+02,2.138905e-03)(-1.070367e+02,2.375449e-03)(-1.043406e+02,2.639458e-03)(-1.016445e+02,2.938510e-03)(-9.894832e+01,3.280443e-03)(-9.625218e+01,3.672973e-03)(-9.355604e+01,4.123339e-03)(-9.085990e+01,4.638007e-03)(-8.816376e+01,5.222463e-03)(-8.546762e+01,5.881101e-03)(-8.277148e+01,6.617242e-03)(-8.007535e+01,7.433240e-03)(-7.737921e+01,8.330706e-03)(-7.468307e+01,9.310778e-03)(-7.198693e+01,1.037444e-02)(-6.929079e+01,1.152283e-02)(-6.659465e+01,1.275747e-02)(-6.389851e+01,1.408048e-02)(-6.120237e+01,1.549455e-02)(-5.850623e+01,1.700284e-02)(-5.581009e+01,1.860870e-02)(-5.311395e+01,2.031523e-02)(-5.041781e+01,2.212463e-02)(-4.772167e+01,2.403753e-02)(-4.502553e+01,2.605217e-02)(-4.232939e+01,2.816366e-02)(-3.963325e+01,3.036323e-02)(-3.693711e+01,3.263761e-02)(-3.424097e+01,3.496864e-02)(-3.154483e+01,3.733310e-02)(-2.884869e+01,3.970279e-02)(-2.615255e+01,4.204501e-02)(-2.345641e+01,4.432325e-02)(-2.076027e+01,4.649822e-02)(-1.806414e+01,4.852913e-02)(-1.536800e+01,5.037513e-02)(-1.267186e+01,5.199693e-02)(-9.975716e+00,5.335835e-02)(-7.279577e+00,5.442788e-02)(-4.583437e+00,5.518007e-02)(-1.887298e+00,5.559666e-02)(8.088419e-01,5.566740e-02)(3.504981e+00,5.539055e-02)(6.201121e+00,5.477293e-02)(8.897261e+00,5.382964e-02)(1.159340e+01,5.258331e-02)(1.428954e+01,5.106316e-02)(1.698568e+01,4.930364e-02)(1.968182e+01,4.734298e-02)(2.237796e+01,4.522157e-02)(2.507410e+01,4.298043e-02)(2.777024e+01,4.065962e-02)(3.046638e+01,3.829690e-02)(3.316252e+01,3.592666e-02)(3.585866e+01,3.357896e-02)(3.855480e+01,3.127909e-02)(4.125094e+01,2.904726e-02)(4.394707e+01,2.689868e-02)(4.664321e+01,2.484390e-02)(4.933935e+01,2.288930e-02)(5.203549e+01,2.103782e-02)(5.473163e+01,1.928969e-02)(5.742777e+01,1.764327e-02)(6.012391e+01,1.609574e-02)(6.282005e+01,1.464379e-02)(6.551619e+01,1.328415e-02)(6.821233e+01,1.201389e-02)(7.090847e+01,1.083068e-02)(7.360461e+01,9.732752e-03)(7.630075e+01,8.718841e-03)(7.899689e+01,7.787939e-03)(8.169303e+01,6.939005e-03)(8.438917e+01,6.170658e-03)(8.708531e+01,5.480877e-03)(8.978145e+01,4.866752e-03)(9.247759e+01,4.324334e-03)(9.517373e+01,3.848573e-03)(9.786987e+01,3.433382e-03)(1.005660e+02,3.071806e-03)(1.032621e+02,2.756280e-03)(1.059583e+02,2.478972e-03)(1.086544e+02,2.232152e-03)(1.113506e+02,2.008577e-03)(1.140467e+02,1.801845e-03)(1.167428e+02,1.606694e-03)(1.194390e+02,1.419205e-03)(1.221351e+02,1.236919e-03)(1.248313e+02,1.058831e-03)(1.275274e+02,8.852784e-04)(1.302235e+02,7.177328e-04)(1.329197e+02,5.585140e-04)(1.356158e+02,4.104498e-04)

};
\addlegendentry{$E_2=0.2$~pJ}

\end{axis}    
\end{tikzpicture}
    \caption{Transmitted $4$-EM MTB pulses over $80$~Km lossless fiber.}
    \label{FigMTB4EM}
\end{figure}
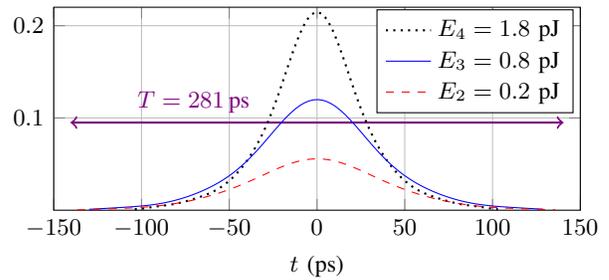

\begin{figure}[t!]
\centering
\input{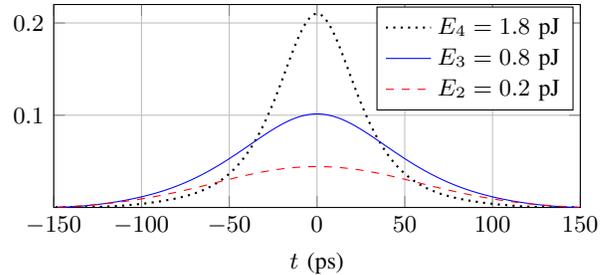}
    \caption{Magnitude of received $4$-EM MTB pulses over $80$~Km lossless fiber.}
    \label{FigRxMTB4EM}
\end{figure}

\section{Conclusions}
\label{Sect:conc}

We introduced waveforms that undergo the minimum time broadening along optical fibers. These waveforms maximize the transmission rates of communication systems that use energy modulation of isolated pulses, and they open the door to finding the maximum achievable rate of such systems. We regard the problem of finding the maximum information theoretic rates as a future research, since solving it requires finding the optimal energy levels and their probabilities based on a given noise model. Another future research is to study the use of these waveforms in intensity modulation direct detection over dispersion-only channels:  pulses with minimum time broadening do not require dispersion compensation and they generate a low number of interfering symbols. We therefore believe that these waveforms could play an important role in complexity-constrained systems. Lastly, the effects of amplifier noise on the designed system are also to be explored. 



\ifCLASSOPTIONcaptionsoff
  \newpage
\fi

\bibliographystyle{ieeetr}
\bibliography{references}

\begin{thebibliography}{10}

\bibitem{AGRAWAL201927}
G.~P. Agrawal, {\em Nonlinear Fiber Optics}.
\newblock Academic Press, 6th~ed., 2019.

\bibitem{yousefi2014Part3}
M.~I. Yousefi and F.~R. Kschischang, ``{Information Transmission Using the
  Nonlinear Fourier Transform, Part III: Spectrum Modulation},'' {\em IEEE
  Transactions on Information Theory}, vol.~60, no.~7, pp.~4346--4369, Apr.
  2014.

\bibitem{hasegawa1973transmission}
A.~Hasegawa and F.~Tappert, ``{Transmission of stationary nonlinear optical
  pulses in dispersive dielectric fibers. I. Anomalous dispersion},'' {\em
  Applied Physics Letters}, vol.~23, no.~3, pp.~142--144, Oct. 1973.

\bibitem{Sander2021}
S.~Wahls, ``Shortening solitons for fiber-optic transmission,'' in {\em 2021
  17th International Symposium on Wireless Communication Systems (ISWCS)},
  Berlin , Germany, Jul. 2021.

\bibitem{span2019time}
A.~Span, V.~Aref, H.~Bülow, and S.~ten Brink, ``{Time-Bandwidth Product
  Perspective for Nonlinear Fourier Transform-Based Multi-Eigenvalue Soliton
  Transmission},'' {\em IEEE Transactions on Communications}, vol.~67, no.~8,
  pp.~5544--5557, Apr. 2019.

\bibitem{chen2020capacity}
Y.~Chen, I.~Tavakkolnia, A.~Alvarado, and M.~Safari, ``{On the Capacity of
  Amplitude Modulated Soliton Communication over Long Haul Fibers},'' {\em
  Entropy}, vol.~22, Aug. 2020.

\bibitem{halpern1979optimum}
P.~Halpern, ``{Optimum finite duration Nyquist signals},'' {\em IEEE
  Transactions on Communications}, vol.~27, no.~6, pp.~884--888, Jun. 1979.

\bibitem{jaffal2019achievable}
Y.~Jaffal and I.~Abou-Faycal, ``{Achievable rates using PAM time-limited pulses
  over band-limited channels: From Nyquist to FTN},'' in {\em 2019 IEEE
  Wireless Communications and Networking Conference (WCNC)}, Marrakesh,
  Morocco, Apr. 2019.

\bibitem{jaffal2022time}
Y.~Jaffal and A.~Alvarado, ``{Pulses with Minimum Residual Intersymbol
  Interference for Faster than Nyquist Signaling},'' {\em arXiv preprint
  arXiv:2203.07156}, Mar. 2022.

\bibitem{lamb1995introductory}
G.~Lamb~Jr, {\em Introductory Applications of Partial Differential Equations:
  With Emphasis on Wave Propagation and Diffusion}.
\newblock John Wiley \& Sons, 1995.

\bibitem{adelman2014software}
R.~Adelman, N.~A. Gumerov, and R.~Duraiswami, ``{Software for Computing the
  Spheroidal Wave Functions Using Arbitrary Precision Arithmetic},'' {\em arXiv
  preprint arXiv:1408.0074}, Aug. 2014.

\end{thebibliography}

\end{document}